\documentclass[a4paper,UKenglish,cleveref, autoref, thm-restate]{lipics-v2021}
\pdfoutput=1 %uncomment to ensure pdflatex processing (mandatatory e.g., to submit to arXiv)
\hideLIPIcs  %uncomment to remove references to LIPIcs series (logo, DOI, ...), e.g., when preparing a pre-final version to be uploaded to arXiv or another public repository

\nolinenumbers

\usepackage{complexity}
\usepackage{mathtools}

\title{How Fast Can We Play Tetris Greedily With Rectangular Pieces?} %TODO Please add

\author{Justin Dallant}{Université libre de Bruxelles, Belgium} {Justin.Dallant@ulb.be}{https://orcid.org/0000-0001-5539-9037}{Supported by the French Community of Belgium via the funding of a FRIA grant.}%TODO mandatory, please use full name; only 1 author per \author macro; first two parameters are mandatory, other parameters can be empty. Please provide at least the name of the affiliation and the country. The full address is optional. Use additional curly braces to indicate the correct name splitting when the last name consists of multiple name parts.

\author{John Iacono}{Université libre de Bruxelles, Belgium} {john@johniacono.com}{https://orcid.org/0000-0001-8885-8172}{Supported by the Fonds de la Recherche Scientifique-FNRS under Grant no MISU F 6001 1.}

\authorrunning{J. Dallant and J. Iacono} %TODO mandatory. First: Use abbreviated first/middle names. Second (only in severe cases): Use first author plus 'et al.'

\Copyright{Justin Dallant and John Iacono} %TODO mandatory, please use full first names. LIPIcs license is "CC-BY";  http://creativecommons.org/licenses/by/3.0/

%TODO mandatory: Please choose ACM 2012 classifications from https://dl.acm.org/ccs/ccs_flat.cfm 

 %TODO mandatory; please add comma-separated list of keywords

\acknowledgements{We thank Nemau for involuntarily providing us with the initial inspiration for this work.}%optional

\EventEditors{}
\EventNoEds{0}
\EventLongTitle{}
\EventShortTitle{}
\EventAcronym{}
\EventYear{}
\EventDate{}
\EventLocation{}
\EventLogo{}
\SeriesVolume{}
\ArticleNo{}
\begin{document}

\maketitle

\begin{abstract}
Consider a variant of Tetris played on a board of width $w$ and infinite height, where the pieces are axis-aligned rectangles of arbitrary integer dimensions, the pieces can only be moved before letting them drop, and a row does not disappear once it is full. Suppose we want to follow a greedy strategy: let each rectangle fall where it will end up the lowest given the current state of the board. To do so, we want a data structure which can always suggest a greedy move. In other words, we want a data structure which maintains a set of $O(n)$ rectangles, supports queries which return where to drop the rectangle, and updates which insert a rectangle dropped at a certain position and return the height of the highest point in the updated set of rectangles. We show via a reduction from the Multiphase problem [P\u{a}tra\c{s}cu, 2010] that on a board of width $w=\Theta(n)$, if the OMv conjecture [Henzinger et al., 2015] is true, then both operations cannot be supported in time $O(n^{1/2-\epsilon})$ simultaneously. The reduction also implies polynomial bounds from the 3-SUM conjecture and the APSP conjecture. On the other hand, we show that there is a data structure supporting both operations in $O(n^{1/2}\log^{3/2}n)$ time on boards of width $n^{O(1)}$, matching the lower bound up to an $n^{o(1)}$ factor.
\end{abstract}

\section{Introduction}
\begin{figure}
    \centering
    \includegraphics[width=\textwidth]{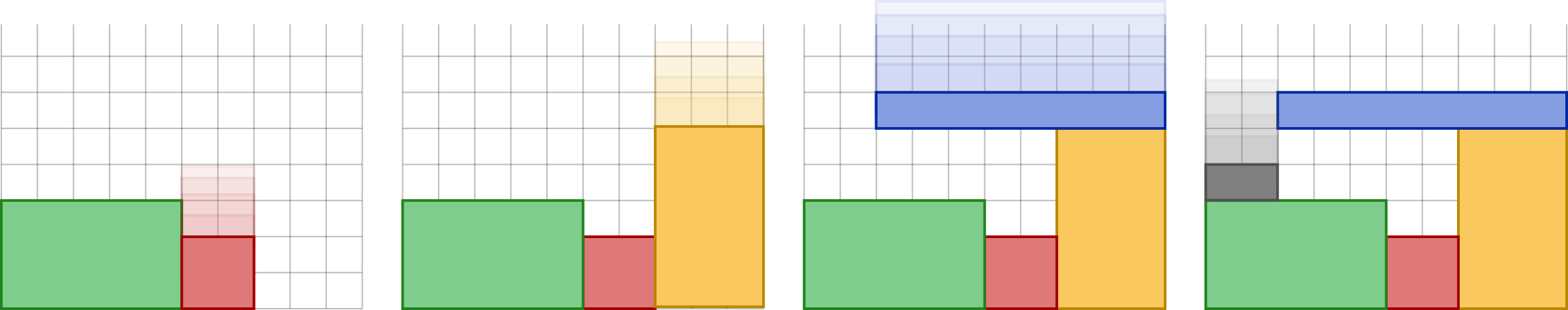}
    \caption{Illustration of a few turns of the game.}\label{fig:game_example}
\end{figure}
Consider a game played on a board of a certain width $w$, which has a wall on the bottom, left and right side but extends on the top as far as we wish. The game goes like this: at each turn, we are given a rectangle of arbitrary integer height and arbitrary integer width (as long as it fits on the board), and we let it drop on the other rectangles (or the bottom of the board) from above starting at the horizontal position of our liking. This is not unlike the basic rule of the famous game Tetris, but instead of having tetrominoes (pieces made of four orthogonally connected unit squares), we have rectangular pieces made of any number of unit squares. See Figure \ref{fig:game_example} for an example of a possible execution of the game.

Now, our goal as the rectangle droppers is to have the rectangles be as low as possible (in Tetris, once a certain height is reached, it means game over!). There are of course various clever strategies and heuristics which could be adopted here. But assume without loss of plausibility that we (the authors) are greedy, and as greedy players we want to adopt the following strategy: at every turn, we let the rectangle drop at a horizontal position which will minimize its vertical position, with no care for the future rectangles we might be given. The question we seek to answer is: how long do we need to think in each turn when adopting this strategy? Here we assume, with a not-so-insignificant loss of plausibility, that we are as powerful as a word-RAM machine. In particular, all conjectures, cited results and claims in this paper are for the word-RAM model with words of length $\Theta(\log n)$, although some carry over to other popular models. All integers appearing in any problem or procedure in this paper have a binary representation of $O(\log n)$ bits.

Let us formalize this problem with the following definition.
\begin{definition}
A Rectangle Dropping Data Structure (RDDS) maintains a set of $O(n)$ independent\footnote{We say two rectangles are independent when their interiors do not intersect.} axis-aligned rectangles in the plane with integer coordinates, lying on or above the $x$-axis and between the vertical lines $x=0$ and $x=w$, and allows the following:
\begin{itemize}
    \item (Preprocessing) Initialize an empty RDDS containing no rectangle.
    \item (Update) Given an axis-aligned rectangle $R$ and a non-negative integer $x_d$, drop $R$ with left border at $x$-coordinate $x_d$ (here we assume that $R$ and $x_d$ are such that $R$ will lie between the lines $x=0$ and $x=w$).
    \item (Query) Given an axis-aligned rectangle $R$, return the position where $R$ must be dropped to end up as low as possible (or one such position if it is not unique) as well as the height of the highest point in the set of rectangles which would result from that move.
\end{itemize}
\end{definition}
Note that an RDDS is actually a bit more powerful than what we need to follow the greedy strategy, as this data structure does not force us to play greedily at each turn, but provides a way to suggest a greedy move at every turn. As no one likes to be forced to do anything (even when the thing in question is what we would have done of our own volition), we will study this variant of the problem. Our results could however be adapted to the case where we can only perform updates corresponding to a greedy strategy.

We would like to give some bounds on the time needed to perform the operations of an RDDS. Because in practice it seems hard to prove good lower bounds on the time needed to solve computational problems, one popular approach (dating back to the seminal works of Cook \cite{DBLP:conf/stoc/Cook71} and Karp \cite{DBLP:conf/coco/Karp72}) has been to show the hardness of some problem assuming the conjectured hardness of some other ``key problem'' (this constitutes what is called a conditional lower bound). This approach, when dealing with specific polynomial bounds (in contrast to showing the $\NP$-hardness of a problem for example) is part of a field of study known today as fine-grained complexity (see e.g., introductory surveys by Bringmann \cite{Bringmann2019} and V.~V.~Williams \cite{Williams2019}). Some of these key problems are the 3-SUM problem, All-Pairs-Shortest-Paths (APSP), Boolean Matrix Multiplication (BMM), Triangle Detection in a graph, Boolean Satisfiability (SAT) and the Orthogonal Vectors problem (2OV).

The use of such approaches for data structures solving dynamic problems was pioneered by P\u{a}tra\c{s}cu \cite{Patrascu2010}. In particular, he introduced the so-called Multiphase problem and showed that it can be fine-grained reduced to many dynamic problems. He also showed that assuming the famous 3-SUM problem is hard, there are polynomial lower bounds to how fast the Multiphase problem can be solved, which in turn implies polynomial lower bounds for the dynamic problems which it reduces to. There has since been much work on conditional lower bounds for dynamic problems \cite{AbboudD2016, AbboudW2014, AbboudWY2018, AlmanMW2017_finegrained, AmirCLL2014_finegrained, AmirKLPPS2019-finegrained, BaswanaCC02016-finegrained, BerkholzKS2017-finegrained, BerkholzKS2018-finegrained, ChenDGWXY2018-finegrained, Dahlgaard2016_finegrained, Henzinger2015, HenzingerL0W2017_finegrained, KarczmarzL2015_finegrained, KopelowitzK2016-finegrained,  Kopelowitz2016, Probst2018-finegrained, BrandNS2019-finegrained, DBLP:conf/focs/WilliamsX20}. We mention some of the expansions on P\u{a}tra\c{s}cu's work relevant to this paper in Section~\ref{section:multiphase}. In a previous paper by the authors \cite{arxiv_geom_lb}, P\u{a}tra\c{s}cu's work and these expansions were exploited to obtain conditional polynomial lower bounds for a variety of dynamic geometric problems. The lower bounds obtained here were inspired by this previous work.

After reviewing the Tetris literature in Section~\ref{s:related}, and the definitions of the hardness conjectures in Section~\ref{section:hardness-conjectures},
 we reduce the Multiphase problem to the conception of an efficient RDDS in Section~\ref{s:reduction}, implying that under the OMv conjecture \cite{Henzinger2015}, any RDDS requires essentially $n^{1/2}$ time per operation. This reduction also implies (weaker, but still polynomial assuming an empty RDDS can be initialized quickly) lower bounds from the 3-SUM conjecture and the APSP conjecture (see Section~\ref{section:multiphase} for definitions of these hardness conjectures). We also give in Section~\ref{s:ds} a data structure matching the lower bound obtained from the OMv conjecture (up to an $n^{o(1)}$ factor).

\subsection{Related work}\label{s:related}

Tetris, created in 1984 by Alexey Leonidovich Pajitnov of the  Dorodnitsyn Computing Centre of the Soviet Academy of Science, was one of the greatest inventions of the USSR, one that was widely embraced by the citizens of both the Warsaw Pact and NATO as well as non-aligned countries.
In the computer science literature, the allure of Tetris has been irrepressible, with works using Tetris found in virtually every subdiscipline.

One might be tempted to begin reviewing the literature with the paper whose title is the one single word ``Tetris''  \cite{DBLP:journals/corr/abs-1811-08614}. 
However, one will quickly realize that the Tetris that they speak of is not the game so beloved by humanity, but rather ``Tetris is an Asynchronous Byzantine Fault Tolerance consensus algorithm designed for next generation high-throughput permission and permissionless blockchain.'' 
In continuing though the literature, one quickly realizes that the name ``Tetris'' has been appropriated mercilessly to entitle and describe things that are not directly about the famed game, but are either meant to invoke the idea of putting pieces together or are just an acronym that the authors found cute, no doubt in the hopes that a small fraction of the shine of the real and true Tetris might befall upon their work \cite{
    DBLP:conf/mig/AdamsonBMSJKT14,
    DBLP:conf/iticse/Black17,
    DBLP:journals/adcm/CasazzaP14,
    DBLP:journals/biosystems/Cowley19,
    DBLP:journals/corr/DoblerR17,
    DBLP:journals/dam/GalliSS01,
    DBLP:conf/asplos/GaoPYHK17,
    DBLP:conf/scopes/GoensKCHSH17,
    DBLP:conf/uss/0001DKBCPH20,
    DBLP:conf/isaac/KutylowskiW97,
    DBLP:journals/tmi/LeePPGS19,
    DBLP:conf/icpp/LiWFHTLL16,
    DBLP:conf/iwqos/LingYWY16,
    DBLP:conf/nana/LiuST21,
    DBLP:conf/iccad/LuWLYL18,
    DBLP:journals/jct/Lupini17,
    DBLP:journals/corr/abs-1811-06841,
    DBLP:conf/gvd/Markl98,
    DBLP:conf/icde/MarklZB99,
    DBLP:conf/sigcse/Schocken17,
    DBLP:conf/mass/StaffordRWCAD17,
    DBLP:conf/eurosys/TumanovZPKHG16,
    DBLP:conf/dac/WestZDFKCW19,
    DBLP:journals/tc/WestZDKFCW20,
    DBLP:conf/chiir/Wilson17,
    DBLP:conf/lctrts/XuT07,
    DBLP:journals/taco/XuT09,
    DBLP:conf/IEEEcloud/ZhangRECSWP11}.
Perhaps the most useful of the Tetris-inspired research has been the creation of a floor-cleaning robot which is hinged and can change shape to various Tetris pieces \cite{DBLP:conf/icra/PrabakaranEPN17,
%DBLP:books/mk/GrayR93,
    DBLP:journals/sensors/LePSM18,
    DBLP:journals/access/PrabakaranEPV18}.

As for the game itself, the most research in the computer science community, unsurprisingly, is about getting computers to play Tetris.
There, Tetris has been used as a research tool more often than any other video game \cite{DBLP:journals/ets/Pinnell15} and has been subjected to just about every machine learning and artificial intelligence paradigm %\cite{DBLP:journals/corr/abs-1905-01652}.
\cite{DBLP:journals/corr/abs-1905-01652,
    DBLP:conf/lwa/BohmKM04,
    DBLP:conf/cig/Boumaza09,
    DBLP:conf/gecco/Boumaza11,
    DBLP:conf/ae/Boumaza11,
    DBLP:conf/bracis/SilvaP17,
    DBLP:conf/nips/GabillonGS13,
    DBLP:conf/gecco/GillespieGS17,
    DBLP:conf/cogsci/GittelsonLSG14,
    DBLP:conf/esann/GrossFS08,
    DBLP:conf/gecco/JaskowskiSLK15,
    DBLP:conf/atal/KnoxS10a,
    DBLP:conf/cec/LangenhovenHE10,
    DBLP:journals/sensors/LePSM18,
    DBLP:conf/icml/LichtenbergS19,
    DBLP:conf/dcai/Marin-LoraCS21,
    DBLP:journals/corr/abs-2004-00377,
    DBLP:conf/cig/Muller-Brockhausen21,
    DBLP:conf/iconip/Phon-Amnuaisuk14,
    DBLP:conf/inns-wc/Phon-Amnuaisuk14,
    DBLP:conf/kes/Phon-Amnuaisuk15,
    DBLP:conf/nafips/PickeringC20,
    DBLP:conf/flairs/RomdhaneL08,
    DBLP:conf/setn/RovatsouL10,
    DBLP:journals/jmlr/Scherrer13,
    DBLP:journals/jmlr/ScherrerGGLG15,
    DBLP:conf/gecco/Schrum18,
    DBLP:conf/iclp/SchullerW15,
    DBLP:conf/cogsci/SibertG17,
    DBLP:conf/cogsci/SibertGL15,
    DBLP:journals/neco/SzitaL06,
    DBLP:conf/gi/Thanh16,
    DBLP:conf/annpr/ThiamKS14,
    DBLP:phd/hal/Thiery10,
    DBLP:journals/icga/ThieryS09,
    DBLP:journals/ria/ThieryS09,
    DBLP:journals/icga/ThieryS09a} 
    including, of course, ant colonies \cite{DBLP:conf/gecco/ChenWWSG09}.
However, success is far from guaranteed, as \emph{The Unsuitability of Supervised Backpropagation Networks for Tetris} shows \cite{DBLP:journals/aans/LewisB15}.

Researchers have discovered that in addition to computers, humans can also play Tetris. They have studied how humans interact as they play Tetris and how this relates to primate socialization in general, observing that  ``baboon social intelligence contrasts both with that of Tetris players and human infants, because both have the ability to provoke epistemic effects'' \cite{DBLP:journals/connection/CowleyM06}.
Other human angles have been studied as well, including 
studying styles of play \cite{DBLP:conf/cogsci/BerryG19,DBLP:conf/hicss/KeatingLTH17}, humanity's suboptimal choices \cite{DBLP:conf/cogsci/SibertG19}, eye movement while playing \cite{DBLP:conf/bcshci/JermannNL10,DBLP:conf/icmi/LiNJ10,
DBLP:conf/mmsp/LiNJ11,DBLP:conf/chiplay/SpielBK19},
and how champion Tetris players play, including the relationship between Hick’s Law and expertise \cite{DBLP:conf/cogsci/LindstedtG20}, and social exertion \cite{DBLP:conf/intetain/MastV16}.
Tetris has been used as well to develop a model of cognition \cite{DBLP:conf/iccm/VekslerG04}.
Additionally, a ``Brain Computer Interface (BCI) game that is inspired [by] the Tetris game'' was created where they perform experiments with children with attention deficit and hyperactivity disorder (ADHD) \cite{DBLP:conf/segah/PiresTCNC11} (see also \cite{DBLP:conf/icmi/KrolFZ17,DBLP:conf/um/PatsisSVT13}).

Several kinds of researchers have worked on helping humans to play Tetris better \cite{DBLP:conf/acg/OikawaHI19} and addressed the all-important work on the influence of the tempo of the music one listens to while playing Tetris on the score
 \cite{DBLP:conf/cig/HufschmittCJ20,DBLP:conf/chiplay/HufschmittCJ20,DBLP:conf/iui/HufschmittCJ21}.
Other researchers worked to study human performance in Tetris, including the elements that can be used to predict whether a human will become a Tetris master \cite{DBLP:conf/cosecivi/Ariza15,
DBLP:conf/cogsci/GrayPLSJMAA14,
DBLP:conf/cogsci/LindstedtG13,
DBLP:conf/cogsci/SibertGL15,
DBLP:conf/cogsci/SibertLG14}.
The issue of dynamically adjusting the difficulty of the game has also been the subject of study \cite{
DBLP:conf/iccbr/ArizaSG17,
DBLP:conf/iccbr/ArizaSG19,
DBLP:conf/cosecivi/LoraSG16,
DBLP:conf/flairs/LoraSGG16,
DBLP:conf/chi/SpielBK17}.

One great step in the development of Tetris was the creation of a mobile version of ``the stable and reliable game based on the mature game algorithm'' \cite{DBLP:conf/icica/WangCL11}. 
Also of note was  \emph{Pop-out Tetris}
\cite{DBLP:conf/jcsse/KongsilpCK18} where they applied the ``Fish Tank Virtual Reality technique to a commodity tablet'' as well as an embedded system created for Tetris \cite{DBLP:conf/csreaESA/PangPT10}. The long history of Tetris and its many implementations has led one group to study these implementations over time, deeming them ``worthy of historical analysis'' \cite{DBLP:conf/digra/Jordan09}.
    
From an education perspective, Tetris was used by \cite{DBLP:conf/digitel/Chen10} to help students learn Chinese proverbs by putting Chinese characters on each Tetris block. 
Within mainstream computer science education, Tetris is a common subject of programming projects
    \cite{DBLP:journals/sigcse/Parlante01}.

 Tetris has also spawned its own fan fiction, which has been studied: ``In [one work] there are also racial and gender complications, the story ending on a humorous note with an absurd twist, as in a Samuel Beckett play, when a message appears over the head of the disillusioned blocks: `Play again?'. It is obvious that in these stories the game of Tetris is an allegory for human life and behavior'' \cite{DBLP:conf/digra/RambuschSEW09}.
 From the point of view of a narrative, ``An experiment was made with the game  Tetris,  showing  that  a  control  system designed  around  evaluation  of  player’s  tension [makes] it possible to obtain an execution trace [close] to a narrative one'' \cite{DBLP:journals/ijigs/DelmasCA08}.

Following the immense success of cannabis legalization, Tetris has followed suit with Enhanced Tetris Legalization 
\cite{DBLP:conf/pci/DadaliarisNKOLS16,
DBLP:conf/seeda/DadaliarisNOHTA16,
DBLP:journals/jolpe/DadaliarisOKNHG17,
DBLP:conf/mocast/OikonomouDLKS18}.

A surprisingly active field has been the ability of a Tetris player to  hide messages via gameplay \cite{DBLP:conf/icmlc/OuC11,DBLP:journals/isci/OuC14,DBLP:journals/iet-ipr/SuCLY20,DBLP:journals/iet-ipr/X21}.

On the more theoretical side,
it has been shown how to view the game through the lens of Krohn-Rhodes theory \cite{DBLP:journals/corr/abs-2004-09022},
how to use mixed integer programming to pack Tetris-like pieces in 3D \cite{DBLP:journals/4or/Fasano04},
 how to create Tetris boards with a given pattern \cite{DBLP:journals/ijigs/HoogeboomK04}, 
 and how to generate polyominoes \cite{DBLP:journals/endm/FormentiM17}.
 In \cite{DBLP:journals/dm/BaccheriniM08}, they
 ``propose a straight Coward algorithm to transform a Tetris-like game into its corresponding automaton.''

 Closer to the topic of this paper is the algorithmic work relating to the complexity of playing Tetris.
 The starting point for this consists of results about the hardness of approximating various Tetris objectives (e.g.~number of rows cleared, number of moves before reaching a certain height), even with foreknowledge of the entire sequence of pieces \cite{DBLP:journals/ijcga/BreukelaarDHHKL04,
    DBLP:journals/corr/cs-CC-0210020,
    DBLP:conf/cocoon/DemaineHL03}.
This work was generalized to show $\NP$-completeness for $k$-ominos with $k\geq 4$, and surprisingly also show that it is $\NP$-complete to clear the board even for $k=2$ when rotations are not allowed \cite{DBLP:journals/jip/DemaineDEHLLY17}.
These reductions were further tightened in \cite{DBLP:journals/jip/AsifCDDHLS20,DBLP:journals/corr/abs-2009-14336} to when the board is small and in \cite{DBLP:journals/corr/Temprano15} for the variant where you can rotate and then move down, but when no rotations are allowed after the block has begun to fall.
The undecidability of whether a set of sequences of Tetris moves described by a regular language contains a sequence that results in an empty board has been proven \cite{DBLP:journals/ipl/HoogeboomK04}.
While this all sounds very negative, a more positive outlook can be found in \emph{Why Most Decisions Are Easy in Tetris} \cite{DBLP:conf/icml/SimsekAK16}.
          
Unrelated to Tetris, the problem studied is a form of dynamic \emph{reverse range query} \cite{DBLP:journals/entcs/Takaoka03}. In a standard range query one has a collection of objects which are preprocessed such that given a range, some question about the objects that intersect the range can be answered. The literature on range queries is immense, with variants depending on the type of objects, the type of ranges, and the specific queries (see, for example, the chapter on range queries in \cite{DBLP:reference/cg/2004}). Classically, what is interesting is that for points and queries that ask how many points intersect the range, the answer is polylog for orthogonal queries and polynomial for non-orthogonal queries, assuming reasonable space.  
In a reverse range query, we are looking for the ranges that would give a certain result, and to provide the extreme such range under some measure; if there is only a single such reverse range query without parameters, this is interesting as a data structure only when the point set is changing dynamically.

For example, the classic problem of finding the smallest circle enclosing a set of points in the plane can be viewed as such a reverse range query, where one is searching for a circular range query that returns $n$ points, and among all such queries is looking for the smallest one; in the dynamic case the best known data structure takes polylogarithmic time, even when deletions are allowed \cite{DBLP:journals/dcg/Chan20a}. 
On the other hand, the best known data structure for maintaining the largest empty circle in a set of points (with center constrained to be inside a fixed triangle) takes time $\tilde{O}(n^{7/8})$ \cite{DBLP:journals/siamcomp/Chan03} (or $\tilde{O}(n^{11/12})$ when deletions are also allowed \cite{DBLP:journals/dcg/Chan20a}). 

For the case of finding (not necessarily axis-aligned) rectangular ranges enclosing all points, the best known data structure takes time $\tilde{O}(n^{1/2})$ to find the minimum-perimeter range and  $\tilde{O}(n^{5/6})$ for the minimum-area range (also \cite{DBLP:journals/siamcomp/Chan03}).
No conditional lower bounds are known for these problems, which remains a major open problem.
In our case we are looking for a three-sided upwards-facing query that returns zero objects,\footnote{Although we have presented our results as a board containing rectangles, our upper and lower bounds also hold for points.} and among all such queries to find the one that is as low as possible.  

\section{A few hardness conjectures}\label{section:hardness-conjectures}

Let us quickly go over some of the key conjectures in fine-grained complexity which will be relevant here. The first, and arguably one of the earliest in the field (with the seminal work of Gajentaan and Overmars \cite{Gajentaan1995}), is the 3-SUM conjecture.
\begin{definition}[3-SUM conjecture]
The following problem (3-SUM) requires $n^{2-o(1)}$ expected time to solve:
given a set of $n$ integers, decide if three of them sum up to $0$.
\end{definition}

The 3-SUM problem can be easily solved in $O(n^2)$ time, and with a lot of effort this can be slightly improved, both for integer and real inputs \cite{DBLP:journals/algorithmica/BaranDP08, DBLP:journals/talg/Chan20, Gronlund2014}.

The second problem we will mention is All-Pairs-Shortest-Path (APSP).
\begin{definition}[APSP conjecture]
The following problem (APSP) requires $n^{3-o(1)}$ expected time to solve:
given a weighted directed graph $G$ with no negative cycles, compute the distance between every pair of vertices in $G$.
\end{definition}
Here, the fastest known algorithm runs in slightly subcubic time \cite{DBLP:journals/siamcomp/Williams18}.
Both problems (3-SUM and APSP) are related in many ways (see e.g., \cite{DBLP:conf/focs/WilliamsX20}). One particular way in which they are related is that they both reduce in a fine-grained manner to the so-called Exact Triangle problem.
In particular, if the following conjecture is false, then both the 3-SUM conjecture and the APSP conjecture are false.
\begin{conjecture}[Exact Triangle conjecture]
The following problem (Exact Triangle) requires $n^{3-o(1)}$ expected time to solve:
given a weighted graph $G$ and a target weight $T$, determine if there is a triangle in $G$ whose edge weights sum to $T$. 
\end{conjecture}

The following problem and the corresponding conjecture were introduced by Henzinger et al. \cite{Henzinger2015} as a way to unify most known conditional lower bounds for dynamic problems under a single conjecture, and sometimes even strengthen the known bounds.
\begin{definition}[Online Boolean Matrix-Vector Multiplication (OMv)]
We are given an $n\times n$ boolean matrix $M$. We can preprocess this matrix, after which we are given a sequence of $n$ boolean column-vectors of size $n$ denoted by $v_1,...,v_n$, one by one. After seeing each vector $v_i$, we must output the product $M v_i$ before seeing $v_{i+1}$.
\end{definition}
\begin{conjecture}[OMv conjecture]
Solving OMv requires $n^{3-o(1)}$ expected time. 
\end{conjecture}

\subsection{P\u{a}tra\c{s}cu's Multiphase Problem}\label{section:multiphase}

The following problem was introduced by P\u{a}tra\c{s}cu \cite{Patrascu2010} as a means to prove polynomial lower bounds for many dynamic problems.
\begin{definition}[Multiphase Problem]
In the Multiphase Problem, we are given a family $\mathcal{F}=\{F_1, \ldots, F_k\}$ of $k$ non-empty subsets of $\{1,2,\ldots, m\}$, such that every element appears in at least one of the subsets. Let $n=m\cdot k$. We want to design a data structure $D$ which maintains a set of $O(n)$ objects and allows for the following:
\begin{itemize}
    \item (Step 1) First, we read $\mathcal{F}$ and are allowed $O(t_1 \cdot n)$ expected time to initialize $D$.
    \item (Step 2) Then, we receive a subset $J \subset \{1,2,\ldots, m\}$ and are allowed $O(t_2 \cdot m)$ expected time to modify $D$.
    \item (Step 3) Finally, we are given an index $1\leq i\leq k$ and after $O(t_3)$ expected time we decide if $J \cap F_{i} \neq \emptyset$.
\end{itemize}
\end{definition}

P\u{a}tra\c{s}cu conjectured that there exist constants $\gamma > 1$ and $\delta > 0$ such that if $k=\Theta(m^\gamma)$ then the Multiphase problem requires $t_1+ t_2+ t_3 \geq n^{\delta}$. While this is still unknown, he showed that this is true if we assume the 3-SUM conjecture. Specifically, he showed the following (taking $k=\Theta(m^{5/2})$).
\begin{theorem}[\cite{Patrascu2010}]
Under the 3-SUM conjecture, the Multiphase Problem requires $t_1+ t_2+ t_3 \geq n^{1/7-o(1)}$.\footnote{Note that here we have expressed the bound in terms of $n=k\cdot m$ whereas P\u{a}tra\c{s}cu originally expressed his bound in terms of what we call $m$ in this paper.}
\end{theorem}
The techniques of P\u{a}tra\c{s}cu also allow to have different trade-offs in the lower bounds for $t_1$, $t_2$ and $t_3$. His results have been expanded upon in different ways. On the one hand, Kopelowitz et al.\cite{Kopelowitz2016} have tightened the reduction from 3-SUM to Set Disjointness (which is an intermediate problem in the reduction from 3-SUM to the Multiphase problem). Their results in particular imply the following:
\begin{theorem}[\cite{Kopelowitz2016}]
Under the 3-SUM conjecture, for any $0<\gamma<1$, the Multiphase Problem requires 
\[t_1 \cdot n + t_2\cdot n + t_3\cdot n^\frac{1+\gamma}{3-2\gamma} \geq n^{\frac{2}{3-2\gamma}-o(1)}.\]
\end{theorem}
On the other hand, the hardness of the Multiphase Problem has been shown from other hardness conjectures. Vassilevska Williams and Xu \cite{DBLP:conf/focs/WilliamsX20} have fine-grained reduced Set Disjointness to the Exact Triangle problem. Their result implies in particular that in the previous theorem, one can replace the 3-SUM conjecture by the Exact Triangle conjecture, giving the following.
\begin{theorem}[\cite{DBLP:conf/focs/WilliamsX20}]
Under the Exact Triangle conjecture (and in particular under either the 3-SUM conjecture or the APSP conjecture), for any $0<\gamma<1$, the Multiphase problem requires 
\[t_1 \cdot n + t_2\cdot n + t_3\cdot n^\frac{1+\gamma}{3-2\gamma} \geq n^{\frac{2}{3-2\gamma}-o(1)}.\]
\end{theorem}

Better bounds can be obtained if we assume the OMv conjecture instead. Indeed, the work of Henzinger et al.\ implies the following.
\begin{theorem}[\cite{Henzinger2015}]\label{thm:multi-omv}
Under the OMv conjecture, for any $0<\gamma<1$, if $t_1$ is at most polynomial in $n$ then the Multiphase Problem requires
\[t_2\cdot n^{1-\gamma} + t_3\cdot n^{\gamma} 
 \geq n^{1-o(1)}.\]
\end{theorem}

\section{Reduction from the Multiphase Problem}\label{s:reduction}

\begin{figure}
    \centering
    \includegraphics[width=\textwidth]{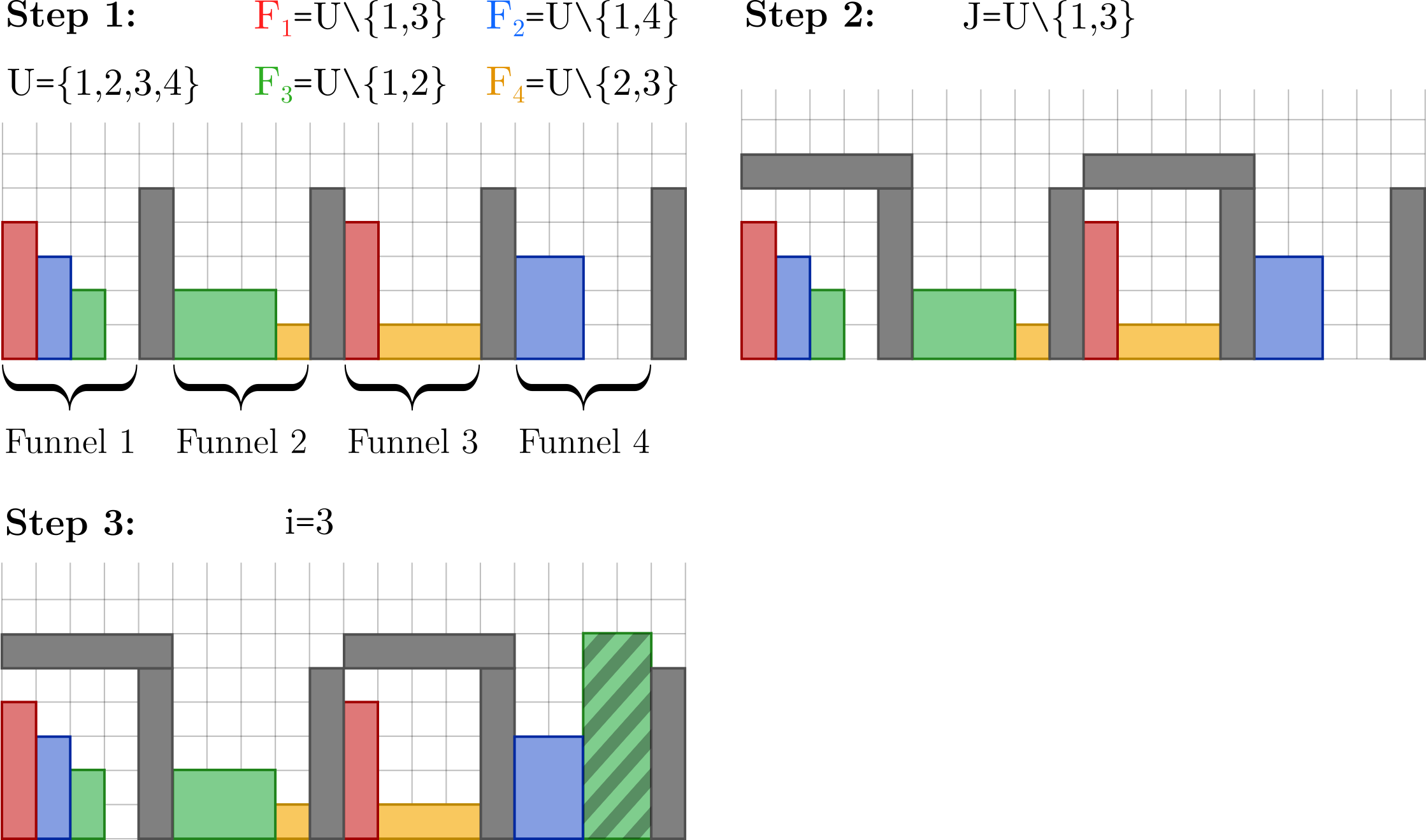}
    \caption{Illustration of the reduction from the Multiphase problem to RDDS. Notice that in Step 3, the striped green rectangle can only be dropped to lie at a height lower than $2$ in a funnel which is not blocked and where there is no other green rectangle. Such a funnel corresponds to an element of $U$ which is both in $J$ and in $F_3$.}\label{fig:phase1}
\end{figure}
\begin{theorem}
If there exists an RDDS which maintains a set of $O(n)$ for boards of width $w=\Theta(n)$ with expected preprocessing time $t_p$, expected update time $t_u$ and expected query time $t_q$, then we can solve the Multiphase Problem with $t_1=t_p/n + t_u$, $t_2 = t_u$ and $t_3=t_q$.
\end{theorem}
\begin{proof}
Suppose we have such an RDDS and consider an instance $\mathcal{F}$ of the Multiphase problem. 

We perform Step 1 as follows (see Figure \ref{fig:phase1}). Given a family $\mathcal{F} = \{F_1,F_2,\ldots,F_k\}$ of subsets of $\{1,2,\ldots,m\}$, we set the width of the board to $m\cdot(k+1) = \Theta(n)$. For every $1\leq j\leq m$, drop a rectangle of height $k+1$ and width $1$ with left border at $x=j(k+1)-1$. This defines $m$ regions of width $k$ on the board (between these rectangles and the borders), which we associate with the $m$ elements of the base set. Now for $j$ going from $1$ to $m$ we create a sort of ``funnel'' structure in region $j$ as follows: for $i$ going from $1$ to $k$, if $j \not\in F_i$, drop a rectangle of height $k+1-i$ such that its right border is at $x=i + (j-1)(k+1)$, its bottom border is at $y=0$ and its left border lies as far left as possible. All in all, we have inserted $n = O(k\cdot m)$ rectangles in $O(t_p + n\cdot t_u)$ expected time.

When given $J \subset \{1,2,\ldots, m\}$ in Step 2, drop a rectangle of width $k+1$ and height $1$ with left border at $x=(j-1)(k+1)$ for every $1\leq j \leq m$ such that $j\not\in J$. This has the effect of ``blocking'' the regions which are not relevant in the next step, and costs $O(m\cdot t_u)$ expected time.

Finally, when given an index $i$ in Step 3, query the RDDS with a rectangle of width $k+1-i$ and height $k+2$. It is easy to check that if there is some $j\in J\cap F_i$, then the bottom of such a rectangle can drop in the funnel corresponding to $j$ in such a way that its bottom lies at or below $y=k-i$ (so its top lies at or below $y=2k+2-i$). On the other hand, if $J\cap F_i=\emptyset$, then the funnel of every region which was not blocked in the previous is too narrow at height $y=k+1-i$ for the rectangle to fall through. Thus, we can solve the Multiphase Problem with a single query in Step 3, costing $O(t_q)$ expected time.
\end{proof}

This then implies the following:
\begin{theorem}
Under the Exact Triangle conjecture (and in particular under either the 3-SUM conjecture or the APSP conjecture), for any $0<\gamma<1$, an RDDS for boards of width $w=\Theta(n)$ requires 
\[t_p + t_u\cdot n + t_q\cdot n^\frac{1+\gamma}{3-2\gamma} \geq n^{\frac{2}{3-2\gamma}-o(1)}.\]
In particular, for $\gamma = 2/3$, we have
\[t_p/n + t_u + t_q \geq n^{1/5-o(1)}.\]

Moreover, if $t_p$ is at most polynomial in $n$, then under the OMv conjecture, for any $0<\gamma<1$ we have
\[t_u\cdot n^{1-\gamma} + t_q\cdot n^{\gamma} \geq n^{1-o(1)}.\]
In particular, for $\gamma = 1/2$, we have
\[t_u + t_q \geq n^{1/2-o(1)}.\]
\end{theorem}

Note that, as mentioned in the introduction of the paper, with some care this reduction could be adapted to the case where a query is always followed by an update dropping the rectangle at the returned position, and no other updates are allowed (i.e. we are forced to always follow the suggested greedy strategy). This can be done by creating ``notches'' of different widths and heights at the top of the rectangles delimiting the different funnels, so that in Step 2 the funnel over which a rectangle is dropped can be chosen by specifying the width of the rectangle. In particular we could get the same bound from the OMv conjecture (exploiting the fact that Theorem \ref{thm:multi-omv} already holds for the Multiphase problem where $k=m$). It could also be adapted to the case where we are allowed to rotate the rectangles before dropping them, by stretching all the rectangles enough horizontally so that rotating them is never advantageous (here we would need to consider boards of larger than linear width). As rectangles are a special case of polyominoes, all these reductions hold also for arbitrary polyominoes of size $O(n)$.

\section{A near optimal data structure}\label{s:ds}

The previous section shows that if we want to follow the greedy strategy, then we cannot always make a move very quickly. While this at least gives us a pretty good excuse for our slow plays, let us now focus on the positive side and prove that we can match this lower bound up to an $n^{o(1)}$ factor.

\begin{figure}
    \centering
    \includegraphics[width=\textwidth]{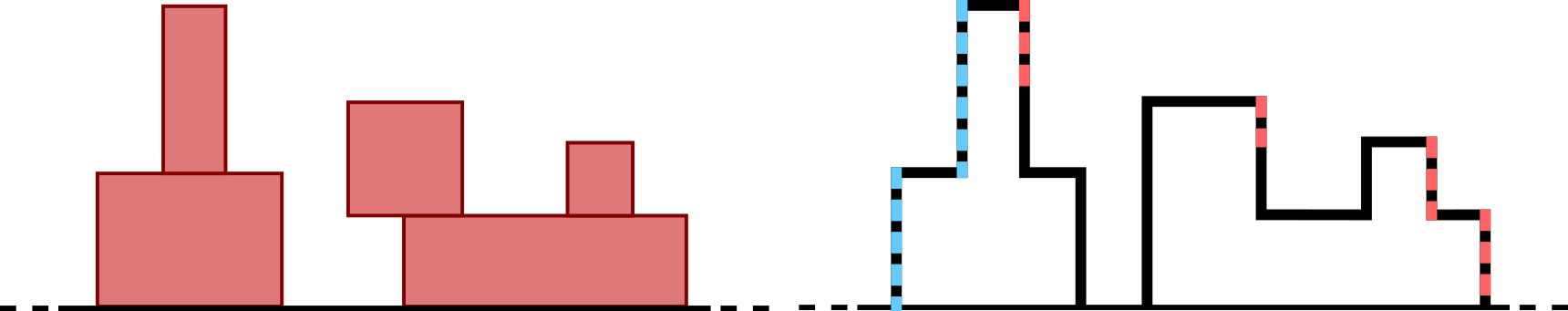}
    \caption{Illustration of the skyline, left staircase (blue dotted line) and right staircase (red dotted line) of a set of rectangles.}\label{fig:skyline}
\end{figure}
\begin{definition}
Given a finite set $S$ of axis-aligned rectangles in the plane lying on or above $x$-axis, the skyline of $S$ is the union of the top, left and right boundaries of the region $B$ obtained as follows:
\begin{itemize}
    \item extend the bottom of every rectangle in $S$ so that its bottom lies on the $x$-axis;
    \item let $B$ be the union of all these newly-obtained rectangles.
\end{itemize}
The left staircase of $S$ consists of the parts of the skyline of $S$ visible from the point at $(-\infty, 0)$. Similarly, the right staircase of $S$ consists of the parts of the skyline of $S$ visible from the point at $(+\infty, 0)$. 
\end{definition}
Notice that starting with a set of $n$ independent axis-aligned rectangles, its skyline is a set of curves consisting of at most $4n-1$ axis-aligned segments, while the left and right staircases each consist of at most $n$ vertical segments. Moreover, the skyline contains all the information needed to know how far we can drop a rectangle of a certain width. 

To get an RDDS with $n^{1/2}$ update and query time, up to a polylogarithmic factor, we will maintain the skyline, left and right staircases of our set of rectangles, broken into at most $O(n^{1/2})$ chunks each consisting of a sequence of roughly $n^{1/2}$ segments appearing consecutively on the skyline.

To start, it is an easy exercise to show the following via a simple plane sweep.

\begin{lemma}\label{l:makesky}
Given a set of $n$ independent axis-aligned rectangles, one can construct its skyline, left and right staircase in $O(n\log n)$ time (where we store the segments appearing in each in order of appearance left to right in three distinct arrays).
\end{lemma}

Now, instead of looking for how low we can drop a rectangle of a given width, let us instead turn this question on its head and ask how wide of a rectangle we can drop if we want it to lie at or below a certain height $h$. If this is impossible, let's say this maximum width is $0$. Remember that we are still constrained by two vertical lines between which our rectangle needs to lie (otherwise, we could have our rectangle be as wide as we wish and simply drop it next to the existing rectangles).
Two special cases arise: when the widest rectangle which can be dropped to height $h$ or below touches the left vertical line or when this rectangle touches the right vertical line.
The maximum width of a rectangle corresponding to these cases can be determined in $O(\log n)$ time given the left and right staircase of the existing rectangles, simply by binary searching through both and finding where they meet the horizontal line $y=h$.
What about the general case? Here we have the following.
\begin{lemma}\label{lemma:log-ds}
Given the skyline of a set of $n$ rectangles, we can construct a data structure in $O(n\log n)$ with which given some height $h$, we can return the width and position of the widest rectangle which can be dropped at or below height $h$ in $O(\log n)$ time.
\end{lemma}

\begin{figure}
    \centering
    \includegraphics[scale=0.6]{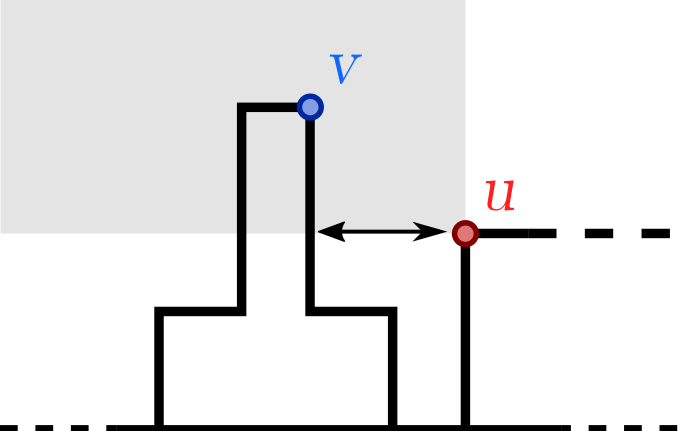}
    \caption{The width of the largest rectangle whose right border touches $u$ is the horizontal distance from $u$ to the horizontally nearest vertex in the shaded region (which in this case is $v$).}\label{fig:dist-neighbor}
\end{figure}

\begin{proof}
We will precompute the answer for all heights $h$ which correspond to some vertex in the skyline. 

We start by computing for every height $h$, the widest rectangle which can be dropped so that its right-side touches a vertex at height $h$ on the skyline. We do this by traversing the skyline from left to right and maintaining the answer for the visited vertices in some data-structure allowing for $O(\log n)$ lookup and updates (some flavor of self-balancing binary search tree for example).
When we reach some vertex $u$ at height $h$, we want to quickly determine how wide we can make a rectangle whose lower-right corner touches $u$ without it intersecting the interior of the skyline. Notice that this is the difference in $x$-coordinates between $u$ and the horizontally nearest neighbor to the left of $u$ which is strictly above $u$ (see Figure \ref{fig:dist-neighbor}). If we store the previously visited vertices in a self-balancing binary search tree sorted by $y$-coordinate and store in every node of that tree the rightmost vertex in the corresponding subtree, we can get this information in $O(\log n)$ time. If this width is larger than the current largest width stored for height $h$, update it. Overall this procedure can be carried out in $O(n\log n)$ time.

Similarly by traversing the skyline from right to left we compute for every height $h$, the widest rectangle which can be dropped so that its left-side touches a vertex at height $h$ on the skyline. 

Then, for every height $h$ corresponding to some vertex, the widest rectangle that can drop at or below height $h$ is the maximum width of a rectangle touching a vertex at height $h'$ for all $h'\leq h$. We can compute this for all relevant $h$, as well as the corresponding location where a rectangle should be dropped, in linear time given the information we have already computed. We store this information in a binary search tree ordered by $h$. Now given any height (not necessarily corresponding to a vertex of the skyline), we can determine in $O(\log n)$ how wide a rectangle can be dropped at or below that height (and where to drop it) by looking in the binary search tree for the largest stored height no larger than $h$.
\end{proof}

We are now ready to break the skyline into chunks to make a data structure which supports insertions.
\begin{theorem}
There is a data structure which maintains a set of $O(n)$ independent axis-aligned rectangles in the plane, lying on or above the $x$ and between the vertical lines $x=0$ and $x=w$, with $O(n\log n)$ construction time and which allows the following operations:
\begin{itemize}
    \item (Update) Given an axis-aligned rectangle $R$ and a number $x_d$, drop $R$ at $x$-coordinate $x_d$ (here we assume that $R$ and $x_d$ are such that $R$ will lie between the lines $x=0$ and $x=w$).
    \item (Query) Given a height $h$ return the widest rectangle which can be dropped to lie at height $h$, together with the position at which it needs to be dropped.
\end{itemize}
Updates can be supported in $O(n^{1/2}\log^{3/2} n)$ time and queries in $O(n^{1/2}\log^{1/2} n)$ time.
\end{theorem}
\begin{proof}
To build our data structure, we start by computing the skyline of the set of rectangles in $O(n\log n)$ time. We add to the skyline long vertical segments corresponding to the lines $x=0$ and $x=w$. Then we break the skyline into $O((n/\log n)^{1/2})$ contiguous pieces, each consisting of a skyline having at most $2(n\log n)^{1/2}$ vertices. We also make sure that the number of vertices in any two consecutive chunks of skyline sums up to at least $(n\log n)^{1/2}$. For each of those, we build the data structure of Lemma \ref{lemma:log-ds}, the left and right staircase and compute the maximum height of a vertex. We store these data structures in a pointer list $L$, in left to right order. This costs a total time of $O((n/\log n)^{1/2} \cdot (n\log n)^{1/2}\log n) = O(n \log n)$.

Now, given a query height $h$, we proceed as follows. We query each of the small data structures in $L$ to find the maximum width of a rectangle which can be dropped at or below height $h$ over all of them. This costs $O((n/\log n)^{1/2}\log n) = O((n\log n)^{1/2})$ time. We are left to deal with rectangles which could span across multiple of these skyline chunks. Here we traverse $L$ with two pointers $p_1$ and $p_2$ in a ``sliding window'' fashion:
\begin{itemize}
    \item We start with $p_1$ at the head of the list and $p_2$ its successor. While both have not reached the end of the list, we go through the following steps.
    \item If the maximum height of a vertex in the skyline chunk corresponding to $p_2$ is at or below $h$ and $p_2$ has not yet reached the tail of the list, we move $p_2$ to its successor in $L$.
    \item Otherwise, we find the widest rectangle which can be dropped at or below height $h$ and overlaps both skylines corresponding to $p_1$ and $p_2$ partially by binary searching through the right staircase corresponding to $p_1$ and the left staircase of $p_2$. Once this is done, we move $p_1$ over to $p_2$ and move $p_2$ to its successor (if $p_2$ has not yet reached the tail of the list).
\end{itemize}
By keeping track of the maximum width and the position it was encountered at throughout this procedure, they can be reported in $O((n/\log n)^{1/2}\log n) = O(n^{1/2}\log^{1/2}n)$ time.

Now, given a rectangle $R$ to drop at a given position, we first find all chunks of skyline in $L$ which will be completely covered by $R$ as well as the pointers $p_1$ and $p_2$ to the elements in $L$ corresponding to the skylines which overlap with the left and right border of $R$. By binary searching through the corresponding staircases and finding the maximum height of a vertex in the chunks which are covered by $R$ we can find out at which height it will end up. Notice that the skylines which are covered by $R$ are no longer relevant at this point, so we can ignore the corresponding data structures in $L$ by moving the successor pointer of $p_1$ so that it points to $p_2$. We now rebuild the data structure $p_1$ points to, including the rectangle $R$, as well as the one $p_2$ points to, excluding all parts of the skyline covered by $R$. If this causes any of the two skylines to have more than $2(n\log n )^{1/2}$ vertices, we can simply divide it in approximate halves and rebuild the two corresponding data structures from scratch in $O((n\log n)^{1/2}\cdot\log n)$ time each. If two consecutive skylines have fewer than $(n\log n)^{1/2}$ vertices together, merge them and rebuild the corresponding data structure from scratch in $O((n\log n)^{1/2}\cdot\log n)$ time. These divisions and merges can happen at most twice for a given update, and ensure that we always have $O((n/\log n)^{1/2})$ small data structures storing $O((n\log n)^{1/2})$ vertices each. The total cost of an update is then $O((n/\log n)^{1/2} + (n\log n)^{1/2}\cdot \log n) = O(n^{1/2}\log^{3/2} n)$.
\end{proof}

Given such a data structure, we can solve our original question. We keep track of all the $O(n)$ possible heights a rectangle could be dropped at, and when given a rectangle of a certain width, we do a binary search through all possible heights to find the lowest height which will accommodate a rectangle of that width. This increases the query time by a factor of $O(\log n)$, giving the following:
\begin{theorem}
There is an RDDS with $O(n\log n)$ construction time which supports queries and updates in $O(n^{1/2}\log^{3/2} n)$ time.
\end{theorem}

This result applies even in the case where the RDDS is not necessarily initialized with an empty set of rectangles. Considering the lower bound from the OMv conjecture given in the previous section, this is likely close to optimal (although we have not made a big effort to optimize the polylogarithmic terms in the runtime). If one were to face an opponent playing the game with a sub-polynomial speed advantage over this data structure, we venture to guess this could be mitigated by an appropriate amount of distraction and foul play. 

\bibliography{tetris}

\end{document}